\documentclass[a4paper,12pt]{article}
\usepackage{fontspec}
\usepackage[english]{babel}
\usepackage{titlesec}
\usepackage[margin=25mm,headheight=15pt]{geometry}
\usepackage{amsmath,amssymb,mathtools}
\usepackage{graphicx}
\usepackage{booktabs}
\usepackage{caption}
\usepackage[section]{placeins}
\usepackage[numbers,sort&compress]{natbib}
\usepackage{microtype}
\usepackage{hyperref}
\usepackage{bookmark}

\hypersetup{
  unicode=true,
  colorlinks=true,
  allcolors=blue,
  pdftitle={Weighted performance bounds for decoded quantum interferometry with imperfect decoding},
  pdfauthor={Maoxin Gao, Yan Chang, Wei Huang, Lili Yan, Zhijian Gou, and Bingjie Xu}
}
\renewcommand{\thesection}{\Roman{section}}
\renewcommand{\thesubsection}{\Alph{subsection}}
\titleformat{\section}{\large\bfseries}{\thesection}{1em}{}
\titleformat{\subsection}{\normalsize\bfseries}{\thesubsection}{1em}{}
\titlespacing*{\section}{0pt}{1.2ex plus .3ex minus .2ex}{.8ex plus .2ex}
\titlespacing*{\subsection}{0pt}{1ex plus .2ex minus .2ex}{.5ex plus .2ex}
\title{\Large Weighted performance bounds for decoded quantum interferometry with imperfect decoding}
\author{%
  Maoxin Gao\textsuperscript{1,2,3},\quad
  Yan Chang\textsuperscript{1,2,3,}%
  \thanks{Corresponding authors are Yan Chang and Bingjie Xu.\quad\href{mailto:cyttkl@cuit.edu.cn}{cyttkl@cuit.edu.cn} (Yan Chang), \mbox{\href{mailto:xbjpku@163.com}{xbjpku@163.com}} (Bingjie Xu)},\quad
  Wei Huang\textsuperscript{4},\\[.3em]
  Lili Yan\textsuperscript{1,2,3},\quad
  Zhijian Gou\textsuperscript{1,2,3},\quad and\quad
  Bingjie Xu\textsuperscript{4,*}\\[.7em]
  \begin{minipage}{.94\textwidth}
    \centering\small
    \textsuperscript{1}School of Cybersecurity (Xin Gu Industrial College),
    Chengdu University of Information Technology, Chengdu 610225, China\\[.15em]
    \textsuperscript{2}Advanced Cryptography and System Security Key Laboratory
    of Sichuan Province, Chengdu 610225, China\\[.15em]
    \textsuperscript{3}SUGON Industrial Control and Security Center,
    Chengdu 610225, China\\[.15em]
    \textsuperscript{4}National Key Laboratory of Security Communication,
    Institute of Southwestern Communication, Chengdu 610041, China
  \end{minipage}%
}
\date{}

\begin{document}
\maketitle

\begin{abstract}
Decoded quantum interferometry (DQI) offers an approach to a class of combinatorial optimization problems by reducing quantum state preparation to a decoding task. However, the existing performance bound for imperfect decoding uses the maximum failure rate over error-weight layers to bound the performance loss uniformly, which may overestimate the effect of decoding failures in layers with low probability. Building on the existing analysis, we combine the failure rate of each layer with its probability weight in the chosen state to derive tighter performance lower bounds. For binary linear constraint problems with uniformly random target values, we obtain a finite-size lower bound on the average satisfaction ratio for fixed state parameters and a fixed decoder, and establish its improvement over the existing bound. In part of the parameter region where the existing bound falls to the random-assignment baseline, the weighted bound still provides a guarantee above that baseline. We further extend the analysis to general finite fields under a model of uniformly random translations of allowed sets of equal size. Trials in individual layers using a fixed belief propagation decoder illustrate the size of the improvement and provide confidence guarantees after accounting for sampling uncertainty. These results reveal the role of retaining the correspondence between the failure distribution and the state probabilities, and suggest directions for more refined performance analyses of other DQI constructions.
\end{abstract}

\section{Introduction}\label{i-ux5f15ux8a00}

Many tasks in practical applications can be formulated as combinatorial optimization problems, including network partitioning, task scheduling, and resource allocation~\citep{abbas2024}. For many of these problems, however, solving large instances exactly often incurs a high computational cost. Obtaining high-quality approximate solutions within a limited computational budget is therefore an important challenge in combinatorial optimization~\citep{abbas2024}. Quantum computing uses superposition and interference to shape the measurement probabilities of candidate solutions, offering the prospect of computational advantages over classical algorithms for certain combinatorial optimization problems~\citep{abbas2024}. The quantum approximate optimization algorithm (QAOA) seeks approximate solutions through alternating evolutions under an objective Hamiltonian and a mixing operator, and has become a widely studied approach to quantum optimization~\citep{farhi2014, farhi2022, boulebnane2022, hastings2019, abbas2024}. Nevertheless, achieving superpolynomial speedups over known classical algorithms for combinatorial optimization remains challenging, motivating the exploration of other quantum algorithms that exploit problem structure~\citep{jordan2025}.

Jordan et al. introduced decoded quantum interferometry (DQI) by exploiting the Fourier structure of the objective function. DQI uses the sparsity of the Fourier spectrum to reduce state preparation for a class of optimization problems to decoding an error-correcting code, thereby translating the available algebraic structure of the problem into structure in the decoding task. By performing decoding coherently and applying a quantum Fourier transform, DQI constructs a quantum state whose amplitudes are determined by a polynomial of the objective function, enhancing the measurement probabilities of high-quality candidate solutions. For polynomial fitting problems with specific algebraic structure, Jordan et al. established superpolynomial speedups over known classical algorithms~\citep{jordan2025}. Subsequent work has developed soft-decoding and quantum-decoding methods~\citep{chailloux2024, piveteau2025, shutty2026, mandalaffine2026, mandal2026}, and explored extensions involving algebraic geometry codes, Hamiltonians, multivariate weighted objectives, non-Hamming geometries, and kernelized constructions~\citep{gu2025, schmidhuber2025, bupauli2026, bumultivariate2026, krajenbrink2026, wang2025}. Related studies address physical noise~\citep{bu2025}, circuits and state preparation~\citep{patamawisut2025, khattar2025, rosmanis2026, wocjan2026}, and software implementations, industrial applications, and benchmarking~\citep{thelen2026, sabater2026, bollmann2026}. For polynomial fitting problems, new existence results and quantum algorithms have further expanded the range of solution quality known to be attainable~\citep{chaillouxopi2025, sun2026, horinaga2026, jo2026}. Other studies examine the scope of DQI through classical sampling, complexity analysis, and limitations for specific problems~\citep{marwaha2025, anschuetz2025, parekh2025, kramer2026, gilfuster2026, kramerbounded2026}.

During DQI state preparation, the decoder processes error vectors of different Hamming weights, where the weight is the number of nonzero components. Increasing the polynomial degree expands the range of error layers involved in preparation and improves the attainable objective value within the regime covered by the ideal-decoding results. A practical decoder, however, may fail to recover errors in the additional layers~\citep{jordan2025}. To relate these decoding failures to optimization performance, Jordan et al. analyzed the resulting performance loss. For MAX-XORSAT with a uniformly random target vector, they derived a lower bound on the expected objective value of the normalized state retained after postselection~\citep{jordan2025}. This relation converts decoding failure rates into guaranteed solution quality, providing a basis for choosing the polynomial degree and decoder settings.

As the polynomial degree approaches the operating threshold of the decoder, failure rates in some high-weight layers may increase substantially. The bound of Jordan et al. uses the maximum failure rate among the participating layers to determine a uniform error penalty~\citep{jordan2025}. Consequently, even high failure rates confined to a few layers may rapidly weaken the guarantee, restricting the parameter range selected on the basis of that bound. However, the maximum failure rate alone leaves the probability weights of these layers unspecified.

For example, with the same ideal state parameters, two failure distributions can have the same maximum failure rate while concentrating their higher failure rates in layers with large and small probability weights, respectively. A uniform treatment based on the worst layer assigns the same error penalty to both distributions, despite the difference in the probability weights of the affected layers. This distinction motivates us to revisit the treatment of decoding errors in the existing bound. The proof of Jordan et al. first estimates the decoding error using the failure rates of individual layers, and then bounds its effect uniformly using the maximum layer failure rate~\citep{jordan2025}. This step simplifies the performance guarantee while leaving room to make further use of the probability weights of the layers.

In this work, we combine the failure rate of each layer with its corresponding probability weight to re-estimate the performance loss caused by decoding failures, obtaining a weighted performance lower bound over the binary field. For the same state parameters and decoder settings, we prove that this bound is at least as strong as the existing bound and give conditions for strict improvement. We further extend the analysis to general finite fields under a model of uniformly random translations of allowed sets of equal size, obtaining corresponding performance guarantees for imperfect decoding. This result also suggests an approach to error analysis in other DQI constructions, namely to examine the guarantees that can be obtained by retaining layerwise failure information in conjunction with the chosen state and optimization objective.

We first review DQI state preparation and the existing performance formula for imperfect decoding. We then derive our unified weighted lower bound, explaining its binary specialization, its conditions of applicability over general finite fields, and its relation to the existing bound. The numerical analysis examines how the failure distribution affects the size of the improvement and accounts for the statistical uncertainty arising from a finite number of trials. Finally, in the context of existing DQI extensions, we discuss the possibility of extending layer-weighted analysis to more general state-preparation problems, and outline directions for using failure distributions to optimize layer amplitudes and balance solution quality against preparation cost.

\section{Preliminaries}\label{ii-dqi-ux7684ux7b97ux6cd5ux57faux7840}

To provide the background for our analysis, we review the DQI framework and performance formulas established by Jordan et al.~\citep{jordan2025}.

\subsection{Linear constraint problems}\label{a-ux7ea6ux675fux95eeux9898ux548cux89e3ux7801ux4efbux52a1}

The maximum linear satisfiability problem (MAX-LINSAT) asks for an assignment satisfying as many linear constraints as possible~\citep{jordan2025}. Let the variables take values in the finite field $\mathbb F_q$, where $q=p^t$, $p$ is prime, and $t$ is a positive integer. Write an assignment as $x\in\mathbb F_q^n$. Given a nonzero coefficient vector $b_i\in\mathbb F_q^n$ and an allowed set $F_i\subset\mathbb F_q$, the $i$th constraint is

\begin{equation*}
b_i\cdot x\in F_i.
\end{equation*}
For $m$ constraints, arranging the vectors $b_i$ as rows gives $B\in\mathbb F_q^{m\times n}$. We consider $|F_i|=r$ with $1\le r<q$.

Let $\mathbf1_{F_i}$ denote the indicator function, and define $f_i(z)=2\mathbf1_{F_i}(z)-1$ and $f(x)=\sum_i f_i(b_i\cdot x)$. If $s(x)$ constraints are satisfied, the number of satisfied constraints and the objective value obey
\begin{equation}
s(x)=\frac{m+f(x)}2.
\label{eq:1}
\end{equation}
Thus $s(x)/m$ is the satisfaction ratio, and its average under a uniformly random assignment is $r/q$.

DQI uses the Fourier coefficients of the centered constraint functions. Write $\bar f=2r/q-1$ and $\phi=2\sqrt{r(q-r)/q}$, and let $g_i(z)=[f_i(z)-\bar f]/\phi$.

For a general finite field, take the additive character
\begin{equation*}
\chi(z)=\exp\!\left(\frac{2\pi i}{p}\operatorname{Tr}_{\mathbb F_q/\mathbb F_p}(z)\right).
\end{equation*}
The Fourier coefficients are defined by~\citep{jordan2025}
\begin{equation}
\widetilde g_i(a)=\frac1{\sqrt q}\sum_{z\in\mathbb F_q}\chi(az)g_i(z).
\label{eq:3}
\end{equation}
Centering gives $\widetilde g_i(0)=0$, and unitarity of the Fourier transform gives $\sum_a|\widetilde g_i(a)|^2=1$. Hence $|\widetilde g_i(a)|^2$ is the probability weight of a nonzero Fourier symbol.

When $q=2$ and $r=1$, the problem reduces to MAX-XORSAT, with single-valued linear constraints over the binary field.

\subsection{DQI state preparation}\label{b-ux4ece-fourier-ux5c42ux5230ux91cfux5b50ux6001}

DQI maps the objective function to quantum state amplitudes~\citep{jordan2025}. Let $P$ be a polynomial of degree at most $\ell$, where $1\le\ell\le m$. The normalized target state is
\begin{equation*}
|P(f)\rangle=\sum_{x\in\mathbb F_q^n}P(f(x))|x\rangle.
\end{equation*}
The probability of measuring $x$ is $|P(f(x))|^2$, so the choice of $P$ adjusts the measurement probabilities of different objective values. Write the Fourier symbols as $y\in\mathbb F_q^m$. The DQI decoder recovers $y$ from the syndrome $B^Ty$. The code generated by $B$ is $C=\{Bx\mid x\in\mathbb F_q^n\}$, with dual code
\begin{equation*}
C^\perp=\{y\in\mathbb F_q^m\mid B^Ty=0\}.
\end{equation*}
The Hamming weight $|y|$ of an error is the number of its nonzero components. The distance of the dual code is the minimum Hamming weight of a nonzero codeword,
\begin{equation}
d^\perp=\min_{\substack{y\in C^\perp\\y\ne0}}|y|.
\label{eq:2}
\end{equation}
If $C^\perp=\{0\}$, we set $d^\perp=\infty$. The condition $2\ell<d^\perp$ ensures that distinct errors of weight at most $\ell$ have distinct syndromes. Since $P$ has degree at most $\ell$, the relevant error weights are $0,\ldots,\ell$. Errors of weight $k$ form layer $k$.

Write $N_k=\binom mk$ and define the product of Fourier coefficients associated with an error $y$ as
\begin{equation*}
G(y)=\prod_{i:y_i\ne0}\widetilde g_i(y_i),
\end{equation*}
where the empty product equals $1$. Normalization of the Fourier coefficients gives $\sum_{|y|=k}|G(y)|^2=N_k$. Over the binary field, each support corresponds to one error.

Let $w_k$ denote the amplitude of layer $k$, and collect these amplitudes in $w=(w_0,\ldots,w_\ell)^T$. With $\sum_{k=0}^{\ell}|w_k|^2=1$, DQI first prepares
\begin{equation}
|\Omega\rangle
=
\sum_{k=0}^{\ell}\frac{w_k}{\sqrt{N_k}}
\sum_{|y|=k}G(y)\,|y\rangle|B^Ty\rangle.
\label{eq:4}
\end{equation}
The input probability of layer $k$ is $|w_k|^2$. The coherent decoder computes a candidate error from the syndrome in the second register and subtracts it from the first register. If the decoder correctly recovers every error of weight at most $\ell$, the first register becomes zero. Omitting this register, the state of layer $k$ is
\begin{equation*}
|\widetilde P^{(k)}\rangle
=
\frac1{\sqrt{N_k}}\sum_{|y|=k}G(y)|B^Ty\rangle.
\end{equation*}
When $2\ell<d^\perp$, distinct syndromes make these layer states orthonormal, and the full remaining state is
\begin{equation}
|\widetilde P(f)\rangle
=
\sum_{k=0}^{\ell}w_k|\widetilde P^{(k)}\rangle.
\label{eq:5}
\end{equation}
Applying the inverse quantum Fourier transform and measuring yields an assignment $x$. Choosing the layer amplitudes $w$ is equivalent to choosing the polynomial $P$~\citep{jordan2025}.

The satisfaction ratio attainable with optimal layer amplitudes is described by the semicircle law of Jordan et al.~\citep{jordan2025}. First consider MAX-XORSAT over the binary field, with $q=2$ and $r=1$, and write $\mu=\ell/m$. Suppose the decoder correctly recovers every error of weight at most $\ell$, and the instance satisfies $2\ell+1<d^\perp$. In the limit $m\to\infty$ with fixed $0<\mu<1/2$, the optimal layer amplitudes give
\begin{equation}
\frac{\langle s\rangle}{m}
\longrightarrow
\frac12+\sqrt{\mu(1-\mu)}.
\label{eq:semicircle}
\end{equation}
The first term is the satisfaction ratio of a random assignment, and the second is the improvement provided by DQI over this baseline. Within the stated interval, the second term increases with $\mu$. Thus, the ability to recover higher-weight errors allows the satisfaction ratio to be increased by raising the polynomial degree.

For general allowed-set sizes, under the same perfect-decoding and distance conditions, Jordan et al. give the asymptotic performance formula~\citep{jordan2025}
\begin{equation}
\frac{\langle s\rangle}{m}
\longrightarrow
\left(
\sqrt{\mu\left(1-\frac rq\right)}
+\sqrt{\frac rq(1-\mu)}
\right)^2,
\qquad \mu\le1-\frac rq.
\label{eq:semicircle-general}
\end{equation}
When $\mu>1-r/q$, the limiting value is $1$. In the balanced case $r/q=1/2$, Eq.~\eqref{eq:semicircle-general} reduces to Eq.~\eqref{eq:semicircle}, whose curve lies on a semicircle. This is the result referred to as the semicircle law by Jordan et al. It relates the error weights that can be handled to the asymptotic satisfaction ratio of DQI.

\subsection{Performance bounds with imperfect decoding}\label{d-ux89e3ux7801ux5931ux8d25ux5982ux4f55ux8fdbux5165ux5df2ux6709ux4e0bux754c}

The semicircle law describes ideal decoding. A practical decoder may fail to recover errors in some weight layers, so failure rates must be incorporated into the performance guarantee. The finite-size analysis of Jordan et al. uses a layer matrix to describe the error terms retained after decoding~\citep{jordan2025}. The entries coupling adjacent layers are~\citep{jordan2025}
\begin{equation*}
a_{k+1}=\sqrt{(k+1)(m-k)},\qquad 0\le k<\ell.
\end{equation*}
Over a general finite field, substitutions between nonzero symbols produce same-layer terms. The diagonal entry for layer $k$ is $kd$, where
\begin{equation*}
d=\frac{q-2r}{\sqrt{r(q-r)}}.
\end{equation*}
The parameter $d$ is determined by $q,r$, whereas $d^\perp$ is determined by the dual code. For single-valued constraints over the binary field, $d=0$.

Arranging these entries in order of weight $0,\ldots,\ell$ gives the real symmetric tridiagonal matrix used in the DQI analysis~\citep{jordan2025}
\begin{equation}
A\equiv A^{(m,\ell,d)}
=
\begin{pmatrix}
0&a_1&&&\\
a_1&d&a_2&&\\
&a_2&2d&\ddots&\\
&&\ddots&\ddots&a_\ell\\
&&&a_\ell&\ell d
\end{pmatrix}.
\label{eq:7}
\end{equation}
All other entries are zero. When $2\ell+1<d^\perp$, this matrix determines the average satisfaction ratio of the full state~\citep{jordan2025},
\begin{equation}
\frac{\langle s\rangle}{m}
=
\frac rq
+
\frac{\sqrt{r(q-r)}}{qm}\,w^\dagger Aw.
\label{eq:6}
\end{equation}
The first term is the random-assignment baseline, and the second contains the layer amplitudes and interactions between layers.

Let $\lambda$ be the largest eigenvalue of $A$. Choosing a corresponding normalized principal eigenvector as the layer amplitudes maximizes the average satisfaction ratio in Eq.~\eqref{eq:6}. For single-valued constraints over the binary field, $d=0$, and
\begin{equation}
\frac{\langle s\rangle}{m}=\frac12+\frac{\lambda}{2m}.
\label{eq:ideal-binary}
\end{equation}
Under the asymptotic conditions of Eq.~\eqref{eq:semicircle}, the spectral analysis of Jordan et al. gives $\lambda/(2m)\to\sqrt{\mu(1-\mu)}$, recovering the semicircle law~\citep{jordan2025}.

Fix $B\in\mathbb F_2^{m\times n}$. Let $v_i$ be the target bit of constraint $i$, which requires $b_i\cdot x=v_i$. Collect these bits in a vector $v$ and take $v$ uniformly distributed over $\mathbb F_2^m$. The corresponding objective function is
\begin{equation}
f_v(x)=\sum_{i=1}^{m}(-1)^{v_i+b_i\cdot x}.
\label{eq:8}
\end{equation}
The number of satisfied constraints is again $s_v(x)=(m+f_v(x))/2$.

With imperfect decoding, we continue to use a normalized real principal eigenvector of the binary layer matrix $A$ as the layer amplitudes $w$. Suppose all weight layers use the same deterministic decoder $\mathsf{Dec}$, which takes a syndrome as input and returns a candidate error or a failure flag. The same syndrome always gives the same output, and both the decoding rule and the layer amplitudes $w$ are independent of $v$.

If the candidate error is incorrect, the subtraction described above leaves a nonzero difference in the first register. DQI discards branches that return a failure flag, measures the first register, and retains the branches with an all-zero outcome. This step postselects on successful decoding, so the retained state contains only correctly recovered errors~\citep{jordan2025}. In layer $k$, these errors form the success set
\begin{equation*}
\mathcal D_k
=
\{y\in\mathbb F_2^m\mid |y|=k,\ \mathsf{Dec}(B^Ty)=y\}.
\end{equation*}
The remaining errors form the failure set
\begin{equation*}
\mathcal F_k
=
\{y\in\mathbb F_2^m\mid |y|=k\}\setminus\mathcal D_k.
\end{equation*}
The failure set includes both failure flags and incorrect candidates. Errors within a binary weight layer have equal probabilities, so the conditional failure rate is
\begin{equation}
\varepsilon_k=\frac{|\mathcal F_k|}{N_k}.
\label{eq:9}
\end{equation}

Write $\mathcal D=\bigcup_{k=0}^{\ell}\mathcal D_k$. Since $G(y)=(-1)^{v\cdot y}$ over the binary field, the unnormalized success state is
\begin{equation}
|\widetilde\psi_{\mathcal D,v}\rangle
=
\sum_{k=0}^{\ell}\frac{w_k}{\sqrt{N_k}}
\sum_{y\in\mathcal D_k}(-1)^{v\cdot y}|B^Ty\rangle.
\label{eq:10}
\end{equation}

For a deterministic decoder, successful errors have distinct syndromes, so the basis states in Eq.~\eqref{eq:10} are orthogonal. Layer $k$ enters preparation with probability $w_k^2$ and is retained with conditional probability $1-\varepsilon_k$. The postselection success probability is therefore $P_{\mathrm{succ}}=1-\bar\varepsilon$, where $\bar\varepsilon=\sum_{k=0}^{\ell}w_k^2\varepsilon_k$~\citep{jordan2025}. We assume $\bar\varepsilon<1$ throughout.

After normalizing the success branch, we apply the inverse quantum Fourier transform and measure. For fixed $v$, averaging over the output probabilities gives the average number of satisfied constraints $\langle s_v\rangle_{\mathcal D}$ and the average objective value $\langle f_v\rangle_{\mathcal D}$. Averaging further over uniformly random $v$, we denote the resulting satisfaction ratio by
\begin{equation*}
R_{\mathcal D}
=\mathbb E_v\frac{\langle s_v\rangle_{\mathcal D}}m
=
\frac12+\frac1{2m}\mathbb E_v\langle f_v\rangle_{\mathcal D}.
\end{equation*}

For single-valued constraints over the binary field, $d=0$. Let $E$ contain the terms removed from the layer matrix because decoding fails at one or both endpoints. Errors differing in one coordinate connect adjacent weight layers. When all errors are decoded successfully, their contributions recover the adjacent-layer entries of $A$.

Averaging over uniformly random $v$ retains only error pairs differing by one bit. The average satisfaction ratio of the success branch is therefore~\citep{jordan2025}
\begin{equation}
R_{\mathcal D}
=\frac12+\frac{w^T(A-E)w}{2m(1-\bar\varepsilon)}.
\label{eq:11}
\end{equation}
Here $w^T(A-E)w$ is the matrix contribution retained between successfully decoded errors, and $1-\bar\varepsilon$ is the postselection success probability. After averaging over the target vector, this expression applies at any dual-code distance.

Counting edges gives the adjacent-layer error bound~\citep{jordan2025}
\begin{equation}
0\le E_{k,k+1}
\le
(\varepsilon_k+\varepsilon_{k+1})a_{k+1}.
\label{eq:12}
\end{equation}
where $0\le k<\ell$. The two terms correspond to the failure rates in the layers at the two endpoints of an edge.

Let $\varepsilon=\max_{0\le k\le\ell}\varepsilon_k$ be the maximum failure rate among the participating layers. From Eq.~\eqref{eq:12}, Jordan et al. obtain $\|E\|_2\le2\varepsilon(m+1)$, yielding the existing performance lower bound~\citep{jordan2025}
\begin{equation}
R_{\mathcal D}\ge R_{\mathrm J}
\equiv\frac12+\frac{\lambda-2\varepsilon(m+1)}{2m(1-\bar\varepsilon)}.
\label{eq:13}
\end{equation}
Here $R_{\mathrm J}$ is the existing lower bound on the satisfaction ratio. The error penalty in the numerator is controlled by the maximum layer failure rate, and the denominator is the postselection success probability.

\section{Weighted performance bounds}\label{iii-ux4fddux7559ux5c42ux4fe1ux606fux7684ux7edfux4e00ux6027ux80fdux754c}

The semicircle law relates larger decodable error weights to higher ideal satisfaction ratios. The imperfect-decoding bound of Jordan et al. further quantifies the effect of errors that fail to decode, providing a basis for selecting the degree when using a practical decoder~\citep{jordan2025}. However, Eq.~\eqref{eq:13} controls the error through the maximum layer failure rate, so layers with lower failure rates also enter the error bound at this maximum value. The layer probabilities $w_k^2$ in the DQI state vary as well, and the layer with the largest failure rate may differ from the layer with the largest probability. Figure~\ref{fig:layer-motivation} shows our belief propagation (BP) trials on the matrix made publicly available by Jordan et al.~\citep{jordan2025}. All layers use the same decoding rule. The observed failure fractions and DQI layer probabilities are compared along a common weight axis.

\begin{figure}[!htb]
\centering
\captionsetup{font=small,skip=5pt}
\includegraphics[width=.88\linewidth]{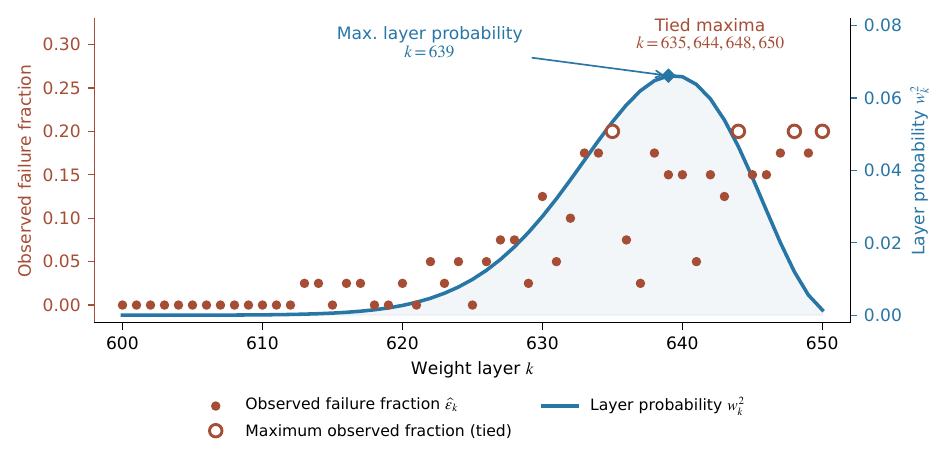}
\caption{Layerwise failure fractions and DQI layer probabilities for a fixed BP decoder. The instance has $m=10000$ constraints and $n=5000$ variables. We take $\ell=650$ and display $600\le k\le650$. All layers use the prior $p=0.065$ and at most $650$ iterations. Red dots show the failure fractions $\widehat\varepsilon_k$ from our $40$ independent trials per layer, read against the left axis. Open red circles mark the tied maximum of $0.20$ at $k=635,644,648,650$, with $8$ failures in each of these layers. The blue curve shows the layer probabilities $w_k^2$ from the normalized principal eigenvector of $A^{(m,\ell,0)}$, read against the right axis. The blue diamond marks their maximum at $k=639$.}
\label{fig:layer-motivation}
\end{figure}
\FloatBarrier

Replacing these differences with a single maximum may overestimate the effect of decoding errors on the chosen state. We therefore start from the existing edgewise estimate in Eq.~\eqref{eq:12} and use the layer amplitudes to bound $w^TEw$ directly. We first establish the weighted lower bound over the binary field, and then apply the same estimate over general finite fields.

\subsection{Weighted lower bound}\label{a-ux542bux5bf9ux89d2ux9879ux7684ux52a0ux6743ux8befux5deeux4f30ux8ba1}\label{c-ux4e8cux8fdbux5236ux6027ux80fdux4fddux8bc1}

When Eq.~\eqref{eq:12} is applied to the quadratic form, each adjacent-layer entry is multiplied by $2w_kw_{k+1}$. Preserving the upper-bound direction requires these products to be nonnegative. In this subsection, we continue to consider single-valued constraints over the binary field. Then $d=0$ and $A$ is nonnegative. The adjacent off-diagonal entries $a_k>0$ connect all weight layers, making $A$ irreducible. By the Perron--Frobenius theorem, its principal eigenvector can be chosen strictly positive~\citep{parlett1998}. We adopt this sign convention below.

\textbf{Theorem 1}\quad Fix $B\in\mathbb F_2^{m\times n}$ and a deterministic decoder independent of the target vector $v$, and let $v$ be uniformly distributed over $\mathbb F_2^m$. Use the normalized principal eigenvector described above as the layer amplitudes. If $\bar\varepsilon<1$, the average satisfaction ratio of the success branch satisfies
\begin{equation}
R_{\mathcal D}
\ge R_{\mathrm w}
\equiv\frac12+\frac{\lambda}{2m}
\frac{1-2\bar\varepsilon}{1-\bar\varepsilon}.
\label{eq:21}
\end{equation}

\textbf{Proof}\quad Record the layer failure rates in the diagonal matrix $\Gamma=\operatorname{diag}(\varepsilon_0,\ldots,\varepsilon_\ell)$, so that $\bar\varepsilon=w^T\Gamma w$. Left and right multiplication by $\Gamma$ introduce the failure rates at the two endpoints of each adjacent-layer entry,
\begin{equation*}
(\Gamma A+A\Gamma)_{k,k+1}
=(\varepsilon_k+\varepsilon_{k+1})a_{k+1}.
\end{equation*}
The diagonal entries of $E$ are zero. Using Eq.~\eqref{eq:12}, $w_kw_{k+1}>0$, and $Aw=\lambda w$, we obtain
\begin{equation}
\begin{aligned}
w^TEw
&\le w^T(\Gamma A+A\Gamma)w\\
&=2w^T\Gamma Aw\\
&=2\lambda w^T\Gamma w\\
&=2\lambda\bar\varepsilon.
\end{aligned}
\label{eq:binary-error}
\end{equation}
The second line uses the symmetry of $A$ and $\Gamma$. Substituting this upper bound into Eq.~\eqref{eq:11} gives
\begin{equation}
\begin{aligned}
R_{\mathcal D}
&=\frac12+\frac{\lambda-w^TEw}{2m(1-\bar\varepsilon)}\\
&\ge\frac12+\frac{\lambda-2\lambda\bar\varepsilon}{2m(1-\bar\varepsilon)}
=R_{\mathrm w}.
\end{aligned}
\label{eq:binary-substitution}
\end{equation}
This proves the theorem.

Equation~\eqref{eq:binary-error} tightens the existing error estimate $2\varepsilon(m+1)$ to $2\lambda\bar\varepsilon$. In summing over layers, we retain the products of adjacent layer amplitudes and then use the principal eigenvalue equation to express the error bound in terms of the weighted failure rate. The correspondence between the failure distribution and the layer probabilities therefore enters the performance lower bound directly. Both lower bounds retain the same postselection denominator and refer to the same output distribution generated by the fixed state and decoder.

Viewed through the finite-size expression of the semicircle law in Eq.~\eqref{eq:ideal-binary}, the weighted lower bound multiplies the ideal gain $\lambda/(2m)$ by $(1-2\bar\varepsilon)/(1-\bar\varepsilon)$. With fully successful decoding, this factor equals $1$, recovering the ideal performance and, in the large-size limit stated above, the semicircle law. When decoding failures occur, this factor provides a performance correction determined by the failure rates of individual layers.

\subsection{Extension to finite fields}\label{b-ux4e00ux822cux6709ux9650ux57dfux4e2dux7684ux9002ux7528ux6761ux4ef6}

Jordan et al. established the performance formula for the full DQI state over general finite fields, given in Eq.~\eqref{eq:6}~\citep{jordan2025}. We retain their Fourier normalization and layer matrix, and extend the analysis to the error set retained after successful decoding. Over $\mathbb F_q$, errors within a layer carry Fourier weights $|G(y)|^2$, and substitutions between nonzero symbols produce same-layer matrix entries. We therefore first derive the average number of satisfied constraints from the success state, and then estimate the adjacent-layer and same-layer errors separately.

In this subsection, $q$ and $r$ are general, and the layer matrix $A$, principal eigenvalue $\lambda$, and layer amplitudes $w$ all correspond to these parameters. Choose fixed allowed sets $F_i^{(0)}$ of size $r$, and take the allowed set of constraint $i$ to be $F_i(v_i)=v_i+F_i^{(0)}$. Denote the corresponding number of satisfied constraints by $s_v(x)$. Throughout this subsection, $v$ is uniform over $\mathbb F_q^m$, and both the decoder and layer amplitudes are independent of $v$. Write the centered function of the template as $g_{i,0}$, and continue to denote its Fourier coefficients and their product by $\widetilde g_i$ and $G$. The translated centered function is $g_{i,v_i}(z)=g_{i,0}(z-v_i)$. A change of variables in Eq.~\eqref{eq:3} gives its Fourier coefficients as
\begin{equation}
\begin{aligned}
\widetilde g_{i,v_i}(a)
&=\frac1{\sqrt q}\sum_{z\in\mathbb F_q}\chi(az)g_{i,0}(z-v_i)\\
&=\frac{\chi(av_i)}{\sqrt q}\sum_{u\in\mathbb F_q}\chi(au)g_{i,0}(u)\\
&=\chi(av_i)\widetilde g_i(a).
\end{aligned}
\label{eq:translation}
\end{equation}
Thus, translation changes the phases while preserving the input error distribution.

The success sets $\mathcal D$ and $\mathcal D_k$ retain their definitions from Sec.~\ref{ii-dqi-ux7684ux7b97ux6cd5ux57faux7840}, with error vectors now taking values in $\mathbb F_q^m$. The failure rate of layer $k$ is computed under the actual Fourier distribution,
\begin{equation}
\varepsilon_k
=\frac1{N_k}\sum_{\substack{|y|=k\\y\notin\mathcal D}}|G(y)|^2.
\label{eq:22}
\end{equation}
For single-valued constraints, the nonzero Fourier coefficients have magnitude $1/\sqrt{q-1}$, so this expression reduces to the failure fraction under the uniform distribution on a weight layer. We continue to use $\bar\varepsilon$ for the total failure probability. The diagonal entries of $\Gamma$ are now the layer failure rates in Eq.~\eqref{eq:22}.

Successful errors have distinct syndromes. By Eq.~\eqref{eq:translation}, the unnormalized success state and its squared norm are
\begin{equation}
\begin{aligned}
|\widetilde\psi_{\mathcal D,v}\rangle
&=\sum_{k=0}^{\ell}\frac{w_k}{\sqrt{N_k}}
  \sum_{y\in\mathcal D_k}G(y)\chi(v\cdot y)|B^Ty\rangle,\\
P_{\mathrm{succ}}
&=\sum_{k=0}^{\ell}\frac{w_k^2}{N_k}
  \sum_{y\in\mathcal D_k}|G(y)|^2
=1-\bar\varepsilon.
\end{aligned}
\label{eq:general-success}
\end{equation}
Again, assume $\bar\varepsilon<1$. The success probability is independent of $v$, so the same normalization denominator can be retained when averaging over instances.

To compute the numerator, we use the Fourier expansion of the operator counting satisfied constraints from Jordan et al.~\citep{jordan2025}. The definition of the centered functions and the inverse transform give
\begin{equation}
s_v(x)=\frac{mr}{q}
+\frac{\sqrt{r(q-r)}}q
\sum_{i=1}^m\sum_{a\ne0}
\widetilde g_i(a)\chi\bigl(a(v_i-b_i\cdot x)\bigr).
\label{eq:satisfaction-fourier}
\end{equation}
In error space, define the translation operator $X_i(a)|y\rangle=|y+ae_i\rangle$, where $e_i$ is the $i$th standard basis vector. The nonconstant component of an individual constraint is represented by $\widetilde g_i(a)X_i(a)$. Summing these operators and organizing the input by weight, write
\begin{equation}
\begin{aligned}
H&=\sum_{i=1}^m\sum_{a\ne0}\widetilde g_i(a)X_i(a),\\
|u_k\rangle&=\frac1{\sqrt{N_k}}\sum_{|y|=k}G(y)|y\rangle,
\qquad U=(|u_0\rangle,\ldots,|u_\ell\rangle).
\end{aligned}
\label{eq:error-space}
\end{equation}
The columns of $U$ are orthonormal. Since the template functions $g_{i,0}$ are real, $\widetilde g_i(-a)=\widetilde g_i(a)^*$, so $H$ is Hermitian. In this error space, the same-layer and adjacent-layer sums described in Sec.~\ref{ii-dqi-ux7684ux7b97ux6cd5ux57faux7840} give $U^\dagger HU=A$~\citep{jordan2025}.

Let $P_{\mathcal D}=\sum_{y\in\mathcal D}|y\rangle\langle y|$ be the projector onto the success set. Apply the inverse Fourier transform to Eq.~\eqref{eq:general-success} to obtain the output amplitudes, and then substitute them into Eq.~\eqref{eq:satisfaction-fourier} to evaluate the expectation. For an error pair $y,y'$ and component $a$ of constraint $i$, the sum over $x$ requires $B^T(y-y'-ae_i)=0$, while the dependence on $v$ is the phase $\chi(v\cdot(y'-y+ae_i))$. Uniform averaging gives
\begin{equation}
\begin{aligned}
\mathbb E_v\chi\bigl(v\cdot(y'-y+ae_i)\bigr)
&=\prod_{j=1}^m\frac1q\sum_{v_j\in\mathbb F_q}
  \chi\bigl(v_j(y'_j-y_j+a\delta_{ij})\bigr)\\
&=\mathbf1_{\{y=y'+ae_i\}}.
\end{aligned}
\label{eq:phase-average}
\end{equation}
Here $\delta_{ij}$ is the Kronecker delta. The retained terms automatically satisfy the corresponding syndrome equality. Thus the normalized instance average is determined by the matrix $M_{\mathcal D}$ associated with the success set,
\begin{equation}
\begin{aligned}
M_{\mathcal D}&=U^\dagger P_{\mathcal D}HP_{\mathcal D}U,
\qquad E=A-M_{\mathcal D},\\
\mathbb E_v\frac{\langle s_v\rangle_{\mathcal D}}m
&=\frac rq+\frac{\sqrt{r(q-r)}}{qm}
  \frac{w^T(A-E)w}{1-\bar\varepsilon}.
\end{aligned}
\label{eq:23}
\end{equation}
We continue to use $E$ for the difference between the full layer matrix and the success matrix. Character orthogonality selects error pairs differing in one coordinate, so $M_{\mathcal D}$ connects equal or adjacent weight layers. This expression for the instance average applies at any dual-code distance.

We next compute the adjacent-layer error. Let $|y|=k$, $y_i=0$, and $y'=y+ae_i$, with $a\ne0$. The weight of the matrix entry for this error pair is
\begin{equation*}
G(y)^*G(y')\widetilde g_i(-a)
=|G(y)|^2|\widetilde g_i(a)|^2
=|G(y')|^2.
\end{equation*}
Hence the adjacent-layer entries are real and nonnegative. Same-layer entries are real because conjugate terms occur in pairs, so $E$ is real symmetric and tridiagonal. Since $\sum_{a\ne0}|\widetilde g_i(a)|^2=1$, the total weight of upward edges from an error in layer $k$ is $(m-k)|G(y)|^2$. Each error in layer $k+1$ has $k+1$ downward edges, each of weight $|G(y')|^2$. Accounting separately for terms removed by failure at either endpoint gives
\begin{equation}
\begin{aligned}
0\le E_{k,k+1}
&\le
\frac{(m-k)N_k\varepsilon_k
 +(k+1)N_{k+1}\varepsilon_{k+1}}
{\sqrt{N_kN_{k+1}}}\\
&=(\varepsilon_k+\varepsilon_{k+1})
  \sqrt{(k+1)(m-k)}.
\end{aligned}
\label{eq:15}
\end{equation}
An edge whose endpoints both fail is counted twice in this upper bound, which can make the estimate loose.

Applying the adjacent-layer bound in Eq.~\eqref{eq:15} to the quadratic form again requires the principal eigenvector components to have the same sign. When $r>q/2$, we have $d<0$ and the layer matrix has negative diagonal entries. The direct application of the Perron--Frobenius theorem in the preceding subsection assumes a nonnegative matrix. Here we can use the tridiagonal structure. The Jacobi-matrix result in Gantmacher--Krein (GK) oscillation theory ensures that, when the adjacent off-diagonal entries are strictly positive, an eigenvector associated with the largest eigenvalue can be chosen strictly positive even for arbitrary real diagonal entries~\citep{gantmacher1950,parlett1998}. We give a direct proof of the positivity statement needed here.

For the chosen real principal eigenvector $w$, let $\widehat w_k=|w_k|$. Taking absolute values preserves the normalization and the diagonal contribution, while the adjacent off-diagonal entries are positive. Hence
\begin{equation}
\begin{aligned}
\lambda
&\ge\widehat w^TA\widehat w\\
&=w^TAw+2\sum_{k=0}^{\ell-1}a_{k+1}
  \bigl(|w_kw_{k+1}|-w_kw_{k+1}\bigr)\\
&\ge w^TAw=\lambda.
\end{aligned}
\label{eq:positive-rayleigh}
\end{equation}
The first line follows from the definition of the maximum Rayleigh quotient. Equality of the endpoints implies that $\widehat w$ also attains the maximum Rayleigh quotient, so $A\widehat w=\lambda\widehat w$. If a component $\widehat w_j=0$, row $j$ gives
\begin{equation*}
\sum_{\substack{0\le k\le\ell\\|k-j|=1}}A_{jk}\widehat w_k
=(\lambda-jd)\widehat w_j=0.
\end{equation*}
The coefficients in the sum are positive and all components are nonnegative, so the neighboring components must also vanish. Iterating along adjacent layers would make every component zero, contradicting normalization. Thus $\widehat w_k>0$. Every nonnegative difference term in Eq.~\eqref{eq:positive-rayleigh} must vanish as well, which shows that neighboring components of the original vector have the same sign. An overall sign choice therefore gives $w_k>0$. This proof holds for every real $d$ and supplies the condition needed for the entrywise bound over general finite fields.

For $q>2$, substitutions between two nonzero symbols preserve the weight. We therefore retain $E_{\mathrm{diag}}=\operatorname{diag}(E_{00},\ldots,E_{\ell\ell})$ separately, with contribution $\Delta_{\mathrm{diag}}=w^TE_{\mathrm{diag}}w$ on the chosen state. By the definition of $M_{\mathcal D}$, this quantity can be computed from the same-layer matrix entries between successful errors. Also write $D=\operatorname{diag}(0,\ldots,\ell)$ and $\bar\eta=w^TD\Gamma w=\sum_k k\varepsilon_kw_k^2$. When $\bar\varepsilon>0$, the ratio $\bar\eta/\bar\varepsilon$ is the mean error weight conditional on failure.

\textbf{Theorem 2}\quad In the random-translation model described above, use the normalized positive principal eigenvector of the layer matrix $A$ as the layer amplitudes. If $\bar\varepsilon<1$, then
\begin{equation}
\mathbb E_v\frac{\langle s_v\rangle_{\mathcal D}}m
\ge
\frac rq+\frac{\sqrt{r(q-r)}}{qm}
\frac{\lambda(1-2\bar\varepsilon)
 +2d\bar\eta-\Delta_{\mathrm{diag}}}
{1-\bar\varepsilon}.
\label{eq:24}
\end{equation}

\textbf{Proof}\quad Let $A_0=A-dD$, extracting the adjacent-layer part of the layer matrix. By the positivity result above, Eq.~\eqref{eq:15} applies term by term to the error quadratic form. Since $A_0w=\lambda w-dDw$, we have
\begin{equation}
\begin{aligned}
w^T(E-E_{\mathrm{diag}})w
&\le w^T(\Gamma A_0+A_0\Gamma)w\\
&=2w^T\Gamma A_0w\\
&=2w^T\Gamma(\lambda w-dDw)\\
&=2\lambda\bar\varepsilon-2d\bar\eta,\\
w^T(A-E)w
&=\lambda-\Delta_{\mathrm{diag}}
  -w^T(E-E_{\mathrm{diag}})w\\
&\ge\lambda(1-2\bar\varepsilon)
 +2d\bar\eta-\Delta_{\mathrm{diag}}.
\end{aligned}
\label{eq:general-error}
\end{equation}
Here we used $D\Gamma=\Gamma D$. Substituting the last line into Eq.~\eqref{eq:23} proves the claim.

The same-layer contribution also admits a conservative estimate in terms of the first moment of the error weight over failed inputs. Let $H_{\mathrm{same}}$ be the weight-preserving part of $H$. On the single-coordinate space, define $X(a)|z\rangle=|z+a\rangle$ and let $K_i=\sum_{a\ne0}\widetilde g_i(a)X(a)$. With $P_*=I-|0\rangle\langle0|$ projecting onto the nonzero-symbol space, $|\gamma_i\rangle=K_i|0\rangle$ is a normalized vector in that space. The single-constraint function $\sqrt q\,g_{i,0}$ takes the two values $\sqrt{(q-r)/r}$ and $-\sqrt{r/(q-r)}$, which satisfy $(\sqrt q\,g_{i,0})^2=1+d\sqrt q\,g_{i,0}$. The Fourier transform maps multiplication by this function to $K_i$, giving
\begin{equation}
\begin{aligned}
K_i^2&=I+dK_i,\\
P_*K_iP_*|\gamma_i\rangle
&=P_*K_i^2|0\rangle\\
&=P_*(|0\rangle+d|\gamma_i\rangle)
=d|\gamma_i\rangle.
\end{aligned}
\label{eq:same-layer-eigenvector}
\end{equation}
On each support of weight $k$, the symbol part of $|u_k\rangle$ is a product of $k$ vectors $|\gamma_i\rangle$. Thus $H_{\mathrm{same}}|u_k\rangle=kd|u_k\rangle$. The same Fourier transform also shows that the smallest eigenvalue of $K_i$ is $-\sqrt{r/(q-r)}$. For any vector $|\xi\rangle$ in the nonzero-symbol space, $\langle\xi|P_*K_iP_*|\xi\rangle=\langle\xi|K_i|\xi\rangle\ge-\sqrt{r/(q-r)}\langle\xi|\xi\rangle$. On a support of weight $k$, the same-layer operator is the sum of these $k$ single-coordinate compressions, so its smallest eigenvalue is at least $-k\sqrt{r/(q-r)}$.

Let $Q_{\mathcal D}=I-P_{\mathcal D}$, and let $|h\rangle=Q_{\mathcal D}Uw$ be the failure component of the input state. Using the same-layer action derived above, we obtain
\begin{equation}
\begin{aligned}
\Delta_{\mathrm{diag}}
&=w^TU^\dagger
  (H_{\mathrm{same}}-P_{\mathcal D}H_{\mathrm{same}}P_{\mathcal D})Uw\\
&=w^TU^\dagger
  (Q_{\mathcal D}H_{\mathrm{same}}
   +H_{\mathrm{same}}Q_{\mathcal D}
   -Q_{\mathcal D}H_{\mathrm{same}}Q_{\mathcal D})Uw\\
&=2d\bar\eta-\langle h|H_{\mathrm{same}}|h\rangle\\
&\le\left(2d+\sqrt{\frac r{q-r}}\right)\bar\eta.
\end{aligned}
\label{eq:diagonal-bound}
\end{equation}
The last step uses the fact that the squared norm of the failure component in layer $k$ is $w_k^2\varepsilon_k$. Substituting this upper bound into Eq.~\eqref{eq:24} yields a lower bound computable from the layer failure rates alone,
\begin{equation}
\mathbb E_v\frac{\langle s_v\rangle_{\mathcal D}}m
\ge
\frac rq+\frac{\sqrt{r(q-r)}}{qm}
\frac{\lambda(1-2\bar\varepsilon)
 -\sqrt{r/(q-r)}\,\bar\eta}
{1-\bar\varepsilon}.
\label{eq:general-conservative}
\end{equation}
When same-layer matrix entries are available, Eq.~\eqref{eq:24} can be used directly to retain more information. For $q=2$ and $r=1$, we have $d=0$ and $H_{\mathrm{same}}=0$, hence $\Delta_{\mathrm{diag}}=0$, and Eq.~\eqref{eq:24} recovers Theorem 1 exactly. Over general finite fields, $d=0$ corresponds to balanced constraints, and same-layer matrix entries still need to be evaluated on the success set.

\section{Comparison of performance bounds}\label{iv-ux6570ux503cux4f30ux4f30}
The weighted estimate in Sec.~\ref{iii-ux4fddux7559ux5c42ux4fe1ux606fux7684ux7edfux4e00ux6027ux80fdux754c} retains the correspondence between layer failure rates and state probabilities. To determine the resulting improvement, we compare performance bounds for the same state parameters and decoder settings, and use trials in individual layers with a fixed BP decoder to examine the size of the improvement and its statistical uncertainty.

\subsection{Theoretical comparison}\label{d-ux540cux4e00ux8bbeux7f6eux4e0bux7684ux65b0ux65e7ux754cux6bd4ux8f83}

Under the random-translation model of Theorem 2, fixing the layer amplitudes $w$ and the decoder also fixes the success probability and same-layer error $\Delta_{\mathrm{diag}}$. To examine the improvement due to layer weighting, we construct a comparison lower bound from Eq.~\eqref{eq:15} using the maximum layer failure rate. Here $\varepsilon=\max_{0\le k\le\ell}\varepsilon_k$, with layer failure rates evaluated using the Fourier weights in Eq.~\eqref{eq:22}.

Continue to write $A_0=A-dD$. Replacing each $\varepsilon_k+\varepsilon_{k+1}$ in the adjacent-layer estimate by the upper bound $2\varepsilon$ and using $w_k>0$ gives
\begin{equation}
\begin{aligned}
w^T(E-E_{\mathrm{diag}})w
&\le 4\varepsilon\sum_{k=0}^{\ell-1}a_{k+1}w_kw_{k+1}\\
&=2\varepsilon w^TA_0w\\
&=2\varepsilon(\lambda-dw^TDw).
\end{aligned}
\label{eq:general-max-error}
\end{equation}
Here $w^TDw=\sum_{k=0}^{\ell}kw_k^2$ is the mean input error weight. Substituting Eq.~\eqref{eq:general-max-error} into the performance expression in Eq.~\eqref{eq:23} gives the satisfaction-ratio lower bound used for comparison,
\begin{equation}
R_{\max}^{(q)}
=\frac rq+\frac{\sqrt{r(q-r)}}{qm}
\frac{\lambda-2\varepsilon(\lambda-dw^TDw)-\Delta_{\mathrm{diag}}}
{1-\bar\varepsilon}.
\label{eq:general-max-bound}
\end{equation}
Denote the right-hand side of Eq.~\eqref{eq:24} in Theorem 2 by $R_{\mathrm w}^{(q)}$. Both lower bounds retain identical same-layer errors and the same success-probability denominator. Their difference therefore isolates the improvement from weighting the adjacent-layer contributions,
\begin{equation}
\begin{aligned}
R_{\mathrm w}^{(q)}-R_{\max}^{(q)}
&=\frac{2\sqrt{r(q-r)}}{qm(1-\bar\varepsilon)}
\bigl[\varepsilon(\lambda-dw^TDw)
      -\lambda\bar\varepsilon+d\bar\eta\bigr]\\
&=\frac{2\sqrt{r(q-r)}}{qm(1-\bar\varepsilon)}
\sum_{k=0}^{\ell}(\varepsilon-\varepsilon_k)(\lambda-kd)w_k^2
\ge0.
\end{aligned}
\label{eq:general-bound-gap}
\end{equation}
The final inequality follows from the principal eigenvalue equation. The adjacent off-diagonal entries of $A_0$ are strictly positive, so $(\lambda-kd)w_k=(A_0w)_k>0$. Together with $\varepsilon\ge\varepsilon_k$, this makes every summand nonnegative. If the layer failure rates differ, at least one term is strictly positive, and the weighted lower bound is strictly tighter. If all layer failure rates are equal, the two estimates coincide. This comparison holds for either sign of $d$.

Equation~\eqref{eq:general-bound-gap} also shows that the improvement depends on the differences between the maximum failure rate and the individual layer failure rates, together with the corresponding weights $(\lambda-kd)w_k^2$. If Eq.~\eqref{eq:diagonal-bound} is used to estimate the same-layer error, replacing $\Delta_{\mathrm{diag}}$ by this upper bound in both expressions leaves their difference unchanged. The weighted lower bound then becomes Eq.~\eqref{eq:general-conservative}.

The difference above determines how much the lower bound improves. The guaranteed solution quality can also be assessed against the random-assignment baseline $r/q$. Since $\bar\varepsilon<1$ and the prefactor is positive, the lower bound in Theorem 2 exceeds this baseline precisely when
\begin{equation}
R_{\mathrm w}^{(q)}>\frac rq
\quad\Longleftrightarrow\quad
\lambda(1-2\bar\varepsilon)+2d\bar\eta-\Delta_{\mathrm{diag}}>0.
\label{eq:general-baseline-condition}
\end{equation}
For Eq.~\eqref{eq:general-conservative}, the corresponding condition is $\lambda(1-2\bar\varepsilon)-\sqrt{r/(q-r)}\,\bar\eta>0$. Thus, over general finite fields, the guarantee depends on both the weighted failure rate and same-layer interactions.

In the binary specialization $q=2$, $r=1$, the same-layer interactions vanish, with $d=0$ and $\Delta_{\mathrm{diag}}=0$. Hence $R_{\mathrm w}^{(q)}$ reduces to $R_{\mathrm w}$ in Theorem 1. We can then also compare directly with the bound of Jordan et al. in Eq.~\eqref{eq:13}~\citep{jordan2025}. To separate the effects of the spectral factor and layer weighting, denote the value of Eq.~\eqref{eq:general-max-bound} in this specialization by $R_{\mathrm{sp}}$. Write the corresponding objective-value lower bounds as $L_\alpha=2m(R_\alpha-1/2)$, where $\alpha\in\{\mathrm J,\mathrm{sp},\mathrm w\}$. Then

\begin{equation}
\begin{aligned}
L_{\mathrm J}&=
\frac{\lambda-2\varepsilon(m+1)}{1-\bar\varepsilon},\\
L_{\mathrm{sp}}&=
\frac{\lambda(1-2\varepsilon)}{1-\bar\varepsilon},\\
L_{\mathrm w}&=
\frac{\lambda(1-2\bar\varepsilon)}{1-\bar\varepsilon}.
\end{aligned}
\label{eq:25}
\end{equation}

Here $L_{\mathrm J}$ is the bound of Jordan et al., $L_{\mathrm{sp}}$ retains the principal eigenvalue while using the maximum layer failure rate, and $L_{\mathrm w}$ further incorporates layer weighting. All three have the same success-probability denominator. Using \(0<\lambda\le m\) in the binary case, proved in Appendix~\ref{app:binary-spectral-bound}, and \(\bar\varepsilon\le\varepsilon\), subtraction gives

\begin{equation}
\begin{aligned}
L_{\mathrm{sp}}-L_{\mathrm J}
&=
\frac{\lambda(1-2\varepsilon)-[\lambda-2\varepsilon(m+1)]}
{1-\bar\varepsilon}\\
&=
\frac{2\varepsilon(m+1-\lambda)}{1-\bar\varepsilon}
\ge0,\\[.5ex]
L_{\mathrm w}-L_{\mathrm{sp}}
&=
\frac{\lambda(1-2\bar\varepsilon)-\lambda(1-2\varepsilon)}
{1-\bar\varepsilon}\\
&=
\frac{2\lambda(\varepsilon-\bar\varepsilon)}{1-\bar\varepsilon}
\ge0.
\end{aligned}
\label{eq:26}
\end{equation}

The first difference comes from replacing \(m+1\) with \(\lambda\) and is strictly positive when \(\varepsilon>0\). The second is the general-field difference in Eq.~\eqref{eq:general-bound-gap} at $d=0$, multiplied by $2m$ to convert it to an objective-value bound. Thus $L_{\mathrm J}\le L_{\mathrm{sp}}\le L_{\mathrm w}$, with layer weighting providing an additional strict improvement whenever the layer failure rates differ. This comparison uses the formula values before clipping. Negative lower bounds on the satisfaction ratio may coincide after taking their maximum with \(0\).

In this binary specialization, Eq.~\eqref{eq:general-baseline-condition} reduces to $\bar\varepsilon<1/2$. For the weighted and existing lower bounds, the respective conditions for guaranteeing a satisfaction ratio above the random-assignment baseline are

\begin{equation}
\begin{aligned}
L_{\mathrm w}>0
&\quad\Longleftrightarrow\quad
\lambda(1-2\bar\varepsilon)>0
\quad\Longleftrightarrow\quad
\bar\varepsilon<\frac12,\\
L_{\mathrm J}>0
&\quad\Longleftrightarrow\quad
\lambda-2\varepsilon(m+1)>0
\quad\Longleftrightarrow\quad
\varepsilon<\frac{\lambda}{2(m+1)}.
\end{aligned}
\label{eq:27}
\end{equation}

The condition of Jordan et al. requires a sufficiently small maximum layer failure rate, whereas the condition for the weighted bound depends on the weighted average of the layer failure rates. Consequently, when higher failure rates are concentrated in a few layers of low probability weight, the same distribution may satisfy both \(\bar\varepsilon<1/2\) and \(\varepsilon\ge\lambda/[2(m+1)]\). In that case, the bound of Jordan et al. has fallen to or below the random-assignment baseline, while the weighted lower bound remains above it.

\subsection{Numerical comparison}

To examine the size of these differences in practical decoding, we perform new BP trials on the constraint matrix made publicly available by Jordan et al., using their decoding implementation~\citep{jordan2025}. This binary instance has $m=10000$ constraints and $n=5000$ variables. All weight layers use the prior $p=0.065$ and at most $650$ iterations. The decoder stops as soon as the candidate error matches the input syndrome, and returns the final candidate if the iteration limit is reached. A trial is counted as successful when the candidate error is identical to the input error. Thus all layers use the same decoder whose output is determined by the syndrome, and their failure probabilities correspond to $\varepsilon_k$ in Theorem 1.

To examine the effect of weighted analysis when the existing bound approaches the random-assignment level, we preselect the degree set $\mathcal L=\{600,\ldots,667\}$ and compare how the bounds change as decoding failures increase. This range covers the transition in the computed existing bound from above to below the random-assignment baseline.

We independently and uniformly sample $40$ fixed-weight errors in each layer $k=0,\ldots,667$, for a total of $26720$ trials. The decoding parameters and sample sizes are fixed before sampling. Let $X_k$ be the number of failures in layer $k$ and $\widehat\varepsilon_k=X_k/40$ the observed failure fraction. For each $\ell\in\mathcal L$, we compute the principal eigenpair of the finite-size layer matrix and substitute the same failure fractions into the three lower bounds in Eq.~\eqref{eq:25}.

\begin{figure}[!htb]
\centering
\includegraphics[width=\linewidth]{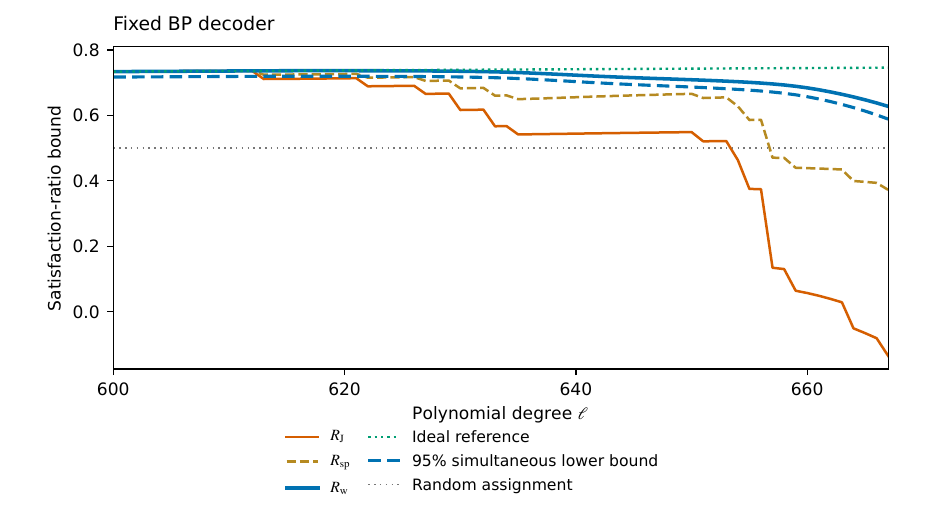}
\caption{Comparison of performance bounds for the same fixed BP decoder. The instance has $m=10000$, $n=5000$, prior $p=0.065$, at most $650$ iterations, and $40$ independent trials per layer. The curves for $R_{\mathrm J}$, $R_{\mathrm{sp}}$, and $R_{\mathrm w}$ are obtained by substituting the observed failure fractions into the corresponding bounds. The blue dashed curve is the lower confidence bound on the average satisfaction ratio from Eq.~\eqref{eq:31}, with a simultaneous confidence level of at least $95\%$ over the $68$ preselected degrees. The green dotted curve is the reference value for zero decoding failures. The gray dotted line is the average satisfaction ratio $1/2$ of a uniformly random assignment, used to assess whether each performance bound guarantees solution quality above this level. Negative values of the lower-bound formulas are retained in the plot.}
\label{fig:partial-win}
\end{figure}

In Fig.~\ref{fig:layer-motivation}, the maximum observed failure fraction and the maximum layer probability occur in different weight layers. At $\ell=650$, the former is $0.200$, while the weighted failure fraction is approximately $0.1225$. The uniform estimate using the maximum and the estimate weighted by layer probabilities therefore give different error penalties. Figure~\ref{fig:partial-win} shows how this difference affects the computed values of the three lower bounds. For $654\le\ell\le667$, the computed value of $R_{\mathrm J}$ is at or below $1/2$, while that of $R_{\mathrm w}$ remains above $1/2$. The intermediate bound $R_{\mathrm{sp}}$ separates the improvements from retaining the principal eigenvalue and from retaining the individual layer failure rates.

These curves use failure fractions from a finite number of trials, so sampling fluctuations must also be taken into account. For example, $8$ failures are observed in layer $650$. The corresponding two-sided $95\%$ binomial confidence interval for the true failure rate of that layer is approximately $[0.091,0.356]$. Since the weighted lower bound depends on $\bar\varepsilon$, we directly control the sampling error of this weighted average. To make the degree dependence explicit, write $e_\ell=\sum_{k=0}^{\ell}w_k^2\varepsilon_k$, which is $\bar\varepsilon$ at that degree. Its estimate is $\widehat e_\ell=\sum_{k=0}^{\ell}w_k^2X_k/T_k$, where $T_k=40$.

Set the total noncoverage probability to $\delta=0.05$ and allocate this budget equally among the $|\mathcal L|=68$ degrees fixed before sampling. An exponential-moment bound for independent bounded variables leads to
\begin{equation}
t_\ell=\sqrt{\frac{\log(|\mathcal L|/\delta)}2
\sum_{k=0}^{\ell}\frac{w_k^4}{T_k}},
\qquad e_{\ell,U}=\min\{1,\widehat e_\ell+t_\ell\}.
\label{eq:weighted-statistical-radius}
\end{equation}
Here $t_\ell$ is determined by the sample sizes and layer weights, and $e_{\ell,U}$ is a one-sided upper confidence bound on the weighted failure rate. The sum of the noncoverage probabilities over the degrees is at most $\delta$, so
\begin{equation}
\Pr\bigl(e_\ell\le e_{\ell,U}
\text{ for all }\ell\in\mathcal L\bigr)\ge1-\delta=0.95.
\label{eq:30}
\end{equation}
This allows degrees within the scan range to be compared and selected while retaining the stated confidence level. Write the weighted lower bound as $F_\ell(e)=1/2+\lambda(1-2e)/[2m(1-e)]$. This function is decreasing on $0\le e<1$. Hence, on the simultaneous coverage event, Theorem 1 gives, whenever $e_{\ell,U}<1$,
\begin{equation}
R_{\mathcal D}\ge F_\ell(e_\ell)\ge F_\ell(e_{\ell,U})
=\frac12+\frac{\lambda}{2m}
\frac{1-2e_{\ell,U}}{1-e_{\ell,U}}.
\label{eq:31}
\end{equation}
Here $R_{\mathcal D}$ is the average satisfaction ratio at the corresponding degree. The uncertainty from sampling individual layers thus propagates directly to the DQI performance guarantee. Proofs of simultaneous coverage and propagation through the monotone bound are given in Appendix~\ref{ux9644ux5f55-d-ux5206ux5c42ux7f6eux4fe1ux754cux5411ux6027ux80fdux4fddux8bc1ux7684ux4f20ux64ad}.

\begin{table}[!htb]
\centering
\captionsetup{font=small}
\caption{Weighted statistical results for the fixed BP decoder. Each layer has $40$ samples, and the total noncoverage probability is $\delta=0.05$ over the $68$ preselected degrees. Here $\widehat e_\ell$ is the weighted average of the observed failure fractions, and $e_{\ell,U}$ is a simultaneous upper confidence bound on the true weighted failure rate. The last two columns give the value obtained by direct substitution and the one-sided lower confidence bound on the average satisfaction ratio, respectively.}
\label{tab:weighted-confidence}
\begin{tabular}{rrrrr}
\toprule
$\ell$ & $\widehat e_\ell$ & $e_{\ell,U}$ & $F_\ell(\widehat e_\ell)$ & $F_\ell(e_{\ell,U})$ \\
\midrule
620 & 0.0039 & 0.0699 & 0.7368 & 0.7199 \\
650 & 0.1225 & 0.1879 & 0.7092 & 0.6869 \\
660 & 0.1982 & 0.2634 & 0.6844 & 0.6573 \\
667 & 0.3247 & 0.3898 & 0.6278 & 0.5889 \\
\bottomrule
\end{tabular}

\end{table}

Table~\ref{tab:weighted-confidence} lists the results for four degrees. At $\ell=650$, we have $\widehat e_\ell\simeq0.1225$ and a statistical error bound of $t_\ell\simeq0.0654$, giving $e_{\ell,U}\simeq0.1879$. Direct substitution into the weighted bound gives $0.7092$, which becomes $0.6869$ after accounting for sampling uncertainty. Throughout the preselected scan, the $95\%$ simultaneous lower confidence bounds for $600\le\ell\le667$ are above the random-assignment baseline $1/2$.

These trials measure the decoding failure probabilities of individual weight layers for a fixed matrix. Theorem 1 then converts them into a guarantee on the satisfaction ratio averaged over uniformly random constraint targets $v$. The three expressions refer to the same state parameters and decoding rule, so the improvement in the comparison is a tightening of the performance guarantee. The output performance of the full quantum circuit requires separate measurement.

\section{Conclusion and discussion}\label{v-ux8ba8ux8bba}

Building on the analysis of Jordan et al., we used the principal eigenvalue equation to derive a weighted lower bound over the binary field that retains the failure rates of individual layers. We extended this bound to general finite fields under a model of uniformly random translations of fixed allowed sets of equal size. For fixed state parameters and a fixed decoder, the theoretical comparison gives the size of the improvement and conditions for strict improvement. New trials with a fixed BP decoder on the public matrix illustrate this improvement and, after accounting for sampling uncertainty, provide simultaneous confidence guarantees over the preselected degree range.

When failure rates differ across layers, the strict improvement from layer weighting reflects the connection between the failure distribution and the state probabilities. Layerwise estimates record where failures occur, while the matrix representation of the objective function and the state amplitudes jointly determine the weights of the error contributions. This suggests starting from the expression for the expected objective value in other DQI constructions to investigate whether retaining this information can tighten their performance bounds.

The depolarizing-noise analysis of Bu et al. provides a concrete starting point. Their imperfect-decoding proof first estimates adjacent-layer errors and then bounds them uniformly by the maximum layer failure rate~\citep{bu2025}. Multivariate DQI partitions constraints into blocks by weight and labels layers by the error weights in each block. It likewise obtains a performance bound using the maximum effective failure rate over adjacent layers~\citep{bumultivariate2026}. These local estimates provide a basis for retaining amplitudes and failure distributions. In the noise model, one must also account for the attenuation of each contribution by noise. In the multivariate construction, one must handle couplings along different block directions and the constraint weights. Weighted summation within these structures is a concrete direction for further tightening performance bounds.

The DQI construction based on algebraic geometry codes by Gu and Jordan retains the finite-field MAX-LINSAT form~\citep{gu2025}. When the random-translation model, equal-size allowed sets, and fixed-decoder conditions of this work hold, the general-field weighted bound can be applied directly to assess the solution quality guaranteed in the presence of partial decoding failures.

A similar matrix structure appears in rank-metric DQI. This construction organizes layers by rank. Under the corresponding dual-distance condition, the expectation of its surrogate objective is represented by a quadratic form of a symmetric tridiagonal matrix~\citep{krajenbrink2026}. This provides a basis for weighted analysis beyond Hamming weight. Such an extension still requires error estimates between rank layers and an analysis of normalization.

Existing studies also establish guarantees through input distributions or quantum state errors. Chailloux and Tillich compute decoding success probabilities under the input error distribution to control the average coset-sampling error~\citep{chailloux2024}. Hamiltonian DQI, developed by Schmidhuber et al., provides fidelity and trace-distance guarantees, and its commuting case includes success terms weighted by layer amplitudes~\citep{schmidhuber2025}. The extension to general Pauli Hamiltonians by Bu et al. bounds the trace-norm error of the state using lower bounds on the success probability for each input~\citep{bupauli2026}. These results motivate investigating error estimates that directly target the expected objective value. In Hamiltonian DQI, this is the energy expectation. The corresponding analysis must also address anticommutation structure, sign interference, and normalization in the noncommuting case.

\label{c-ux89e3ux7801ux9608ux503cux9644ux8fd1ux7684ux53c2ux6570ux9009ux62e9}
Retaining the failure distribution also informs the choice of state parameters. Here we use the ideal principal eigenvector as the layer amplitudes. If the layer failure rates are known, they could be incorporated into amplitude optimization, with the postselection success probability used to assess solution quality and preparation cost. Changing the amplitudes requires a new analysis of the performance guarantee. Trials in individual layers can supply data for this purpose, and can also compare the lower bounds supported by different decoding strategies in a given DQI setting.

\appendix
\renewcommand{\thesection}{\Alph{section}}
\renewcommand{\thesubsection}{\Alph{section}.\arabic{subsection}}
\titleformat{\section}{\large\bfseries}{Appendix \thesection}{1em}{}

\section{Proof of the principal eigenvalue bound}\label{ux9644ux5f55-c-ux4e8cux8fdbux5236ux4e3bux7279ux5f81ux503cux7684ux6709ux9650ux89c4ux6a21ux4e0aux754c}\label{app:binary-spectral-bound}

To prove \(0<\lambda\le m\), as used in Eq.~\eqref{eq:26}, specialize Eq.~\eqref{eq:error-space} to the binary field with templates $F_i^{(0)}=\{0\}$. Then $d=0$, $G(y)=1$, and $H=\sum_iX_i(1)$ is a sum of single-bit flip operators. Since $\|X_i(1)\|_2=1$, the triangle inequality gives $\|H\|_2\le m$.

The layer states are now $|u_k\rangle=N_k^{-1/2}\sum_{|y|=k}|y\rangle$. A single-bit flip connects adjacent layers, and

\begin{equation}
\begin{aligned}
\langle u_{k+1}|H|u_k\rangle
&=
\frac{(m-k)N_k}{\sqrt{N_kN_{k+1}}}\\
&=
\sqrt{(k+1)(m-k)}.
\end{aligned}
\tag{A1}
\label{eq:C1}
\end{equation}

Thus $A$ is the compression of $H$ to the first $\ell+1$ layer states. Every normalized \(z\in\mathbb R^{\ell+1}\) satisfies

\begin{equation}
z^TAz
=
\left\langle\sum_{k=0}^{\ell}z_ku_k
\middle|H\middle|
\sum_{k=0}^{\ell}z_ku_k\right\rangle
\le\|H\|_2
\le m.
\tag{A2}
\label{eq:C2}
\end{equation}

Taking the maximum Rayleigh quotient gives \(\lambda\le m\). Taking \(z=(1,1,0,\ldots,0)^T/\sqrt2\) also gives

\begin{equation*}
\lambda\ge z^TAz
=2z_0z_1A_{01}
=\sqrt m>0.
\end{equation*}

This proves the spectral bound for \(1\le\ell\le m\).

\section{Proof of the statistical guarantees}\label{ux9644ux5f55-d-ux5206ux5c42ux7f6eux4fe1ux754cux5411ux6027ux80fdux4fddux8bc1ux7684ux4f20ux64ad}

This appendix proves Eqs.~\eqref{eq:weighted-statistical-radius}--\eqref{eq:31}. All weight layers use a common syndrome decoder fixed before sampling, with layer failure rates $\varepsilon_k$. For $0\le k\le\max\mathcal L$, the predetermined sample size is $T_k>0$, and all errors are sampled independently and uniformly within their respective weight layers. The degree set $\mathcal L$ is fixed in advance. The layer amplitudes are determined by $m$ and $\ell$, independently of the trial outcomes.

\subsection{Proof of Eq.~(\ref{eq:30})}\label{d.1-ux4e8cux9879ux5206ux5e03ux4e0aux754cux7684ux8054ux5408ux8986ux76d6}

Let $Z_{k,j}$ be the failure indicator for trial $j$ in layer $k$. It takes values $0$ or $1$, has expectation $\varepsilon_k$, and satisfies $X_k=\sum_{j=1}^{T_k}Z_{k,j}$. Fix a degree $\ell$ and let $c_k=w_k^2/T_k$. Then
\begin{equation*}
e_\ell-\widehat e_\ell
=\sum_{k=0}^{\ell}\sum_{j=1}^{T_k}c_k(\varepsilon_k-Z_{k,j}).
\end{equation*}
Each summand lies in an interval of length $c_k$. To control the upper tail of this weighted sum, define the log moment-generating function
$g_{k,j}(s)=\log\mathbb E\exp[s c_k(\varepsilon_k-Z_{k,j})]$.
We have $g_{k,j}(0)=g'_{k,j}(0)=0$. Its second derivative is the variance of $c_kZ_{k,j}$ under the exponentially tilted distribution. This distribution still has two support points, so $g''_{k,j}(s)\le c_k^2/4$. Integrating twice for $s\ge0$ yields
\begin{equation*}
g_{k,j}(s)
=\int_0^s(s-u)g''_{k,j}(u)\,\mathrm du
\le \frac{s^2c_k^2}{8}.
\end{equation*}
By independence of the trials, the exponential moments multiply. Writing $V_\ell=\sum_{k=0}^{\ell}w_k^4/T_k$, we obtain
\begin{equation*}
\begin{aligned}
\mathbb E\exp[s(e_\ell-\widehat e_\ell)]
&=\prod_{k=0}^{\ell}\prod_{j=1}^{T_k}\exp[g_{k,j}(s)]\\
&\le\exp\!\left(\frac{s^2}{8}\sum_{k=0}^{\ell}T_kc_k^2\right)
=\exp\!\left(\frac{s^2V_\ell}{8}\right).
\end{aligned}
\end{equation*}
For any $t>0$, Markov's inequality gives
\begin{equation}
\begin{aligned}
\Pr(e_\ell-\widehat e_\ell\ge t)
&\le \exp(-st)\mathbb E\exp[s(e_\ell-\widehat e_\ell)]\\
&\le\exp\!\left(-st+\frac{s^2V_\ell}{8}\right)\\
&=\exp\!\left(-\frac{2t^2}{V_\ell}\right),
\qquad s=\frac{4t}{V_\ell}.
\end{aligned}
\tag{B1}
\label{eq:D1}
\end{equation}
Taking $t=t_\ell=\sqrt{V_\ell\log(|\mathcal L|/\delta)/2}$ makes the right-hand side equal to $\delta/|\mathcal L|$. Applying a union bound over the preselected degree set $\mathcal L$ gives
\begin{equation}
\begin{aligned}
\Pr\bigl(\exists\,\ell\in\mathcal L,\
e_\ell>\widehat e_\ell+t_\ell\bigr)
&\le\sum_{\ell\in\mathcal L}
\Pr(e_\ell-\widehat e_\ell>t_\ell)\\
&\le\sum_{\ell\in\mathcal L}\frac{\delta}{|\mathcal L|}
=\delta.
\end{aligned}
\tag{B2}
\label{eq:D2}
\end{equation}
Since $e_\ell\le1$, truncating the upper bound to $e_{\ell,U}=\min\{1,\widehat e_\ell+t_\ell\}$ preserves coverage, establishing Eq.~\eqref{eq:30}. The same trials can be reused for different degrees, since the union bound applies irrespective of dependence between their estimates.

\subsection{Proof of Eq.~(\ref{eq:31})}\label{d.2-ux52a0ux6743ux5931ux8d25ux7387ux7684ux7f6eux4fe1ux4f20ux64ad}

On the simultaneous coverage event, $e_\ell\le e_{\ell,U}$. For the binary weighted expression
$F_\ell(e)=1/2+\lambda(1-2e)/[2m(1-e)]$,
differentiation on $0\le e<1$ and $\lambda>0$ give
\begin{equation}
\begin{aligned}
F'_\ell(e)
&=\frac{\lambda}{2m}
\frac{(-2)(1-e)-(1-2e)(-1)}{(1-e)^2}\\
&=-\frac{\lambda}{2m(1-e)^2}<0.
\end{aligned}
\tag{B3}
\label{eq:D3}
\end{equation}
All weight layers use the same fixed decoder, so $e_\ell=\bar\varepsilon$. Combining this with Theorem 1, whenever $e_{\ell,U}<1$ we have
\begin{equation}
R_{\mathcal D}\ge F_\ell(e_\ell)
\ge F_\ell(e_{\ell,U})
=\frac12+\frac{\lambda}{2m}
\frac{1-2e_{\ell,U}}{1-e_{\ell,U}}.
\tag{B4}
\label{eq:D4}
\end{equation}
Together with Eq.~\eqref{eq:D2}, this establishes the simultaneous confidence level for Eq.~\eqref{eq:31}. If $e_{\ell,U}=1$, the trivial satisfaction-ratio bound $R_{\mathcal D}\ge0$ can be used.

\end{document}